\documentclass[12pt,a4paper]{article}
\newif\ifshort
\shorttrue

\usepackage{a4wide}
\usepackage{amsmath}
\usepackage{amssymb}
\usepackage{amsthm}
\usepackage{graphicx}
\usepackage{listings}

\def\Reals{\ensuremath{\mathbb{R}}}

\newtheorem{definition}{Definition}
\newtheorem{lemma}[definition]{Lemma}

\newtheorem{theorem}[definition]{Theorem}

\newcommand{\N}{\mathbb{N}}
\newcommand{\HH}{\mathcal{H}}
\newcommand{\depth}{\operatorname{depth}}
\newcommand{\vx}{\vec{x}}

\newcommand{\R}{\mathbb{R}}
\newcommand{\Rd}{\mathbb{R}^d}

\def\annotation#1#2{%
  \leavevmode\vbox to0pt{%
    \vss
    \cbstart
    \color{red}
    \rlap{\vrule\raise .75em%
       \hbox{\underbar{\normalfont\tiny#1 says: #2}}}\cbend}}

\newcommand{\comment}[1]{}

\renewcommand{\annotation}[2]{}

\title{The parameterized complexity of some geometric problems in unbounded dimension%
}
\author{
 Panos Giannopoulos%
  \thanks{Institut f{\"ur} Informatik, Freie Universit{\"a}t Berlin,
    Takustra{\ss}e 9, D-14195 Berlin, Germany, \{panos, knauer, rote\}@inf.fu-berlin.de.}
  \footnote{This research was
    supported by the German Science Foundation (DFG) under grant Kn~591/3-1.}
  \and
  Christian Knauer\footnotemark[1] \footnotemark[2]
  \and
  G{\"u}nter Rote\footnotemark[1]
}

\begin{document}
\date{}
\maketitle


\begin{abstract}
We study the parameterized complexity of the following fundamental 
geometric problems with respect to the dimension $d$:
\begin{itemize}
\item [i)] Given $n$ points in $\Rd$, compute their minimum enclosing cylinder. 
\item [ii)]  Given two $n$-point sets in $\Rd$, decide whether they can be separated by two hyperplanes.
\item [iii)] Given a system of $n$ linear inequalities with $d$ variables, find a maximum-size feasible subsystem.
\end{itemize}
We show that (the decision versions of) all these problems are W[1]-hard when parameterized by the dimension $d$. 
Our reductions also give a $n^{\Omega(d)}$-time lower bound (under the Exponential Time Hypothesis).

\medskip

\end{abstract}
\section{Introduction}
We study the parameterized complexity of the following three fundamental geometric problems with respect to the dimension 
of the underlying space: minimum enclosing cylinder of a set of points in $\Rd$, 
$2$-linear separation of two point sets in $R^d$, 
and maximum-size feasible subsystem of a system of linear inequalities with $d$ variables. 
All these problems are NP-hard when the dimension $d$ is unbounded 
and all known exact algorithms run in $n^{O(d)}$ time (basically, using brute force), where $n$ is the total 
number of objects in the input sets. 
As with many other geometric problems in $d$ dimensions, it is widely conjectured that the dependence on $d$ cannot be 
removed from the exponent of $n$. However, no evidence of this has been given so far.

In terms of parameterized complexity theory
the question is whether any of these problems is fixed-parameter tractable 
with respect to $d$, i.\,e., whether there exists an algorithm that runs in 
${O}(f(d)n^c)$ time, for some computable function $f$ and some constant $c$ independent of $d$. 
Proving a problem to be W[1]-hard with respect to $d$, gives a strong evidence that such an algorithm 
is not possible, under standard complexity theoretic assumptions. We summarize our results bellow.

\paragraph{Results.}
We study the following decision problems:
\begin{itemize}
\item [i)] Given $n$ unit balls $\Rd$, decide whether there is a line that stabs all the balls. 
           (Note that since the balls are unit, this is the decision version of the problem of computing 
           the minimum enclosing cylinder of a set of $n$ points.) 
\item [ii)]  Given two $n$-point sets in $\Rd$, decide whether they can be separated by two hyperplanes.
\item [iii)] Given a system of $n$ linear inequalities with $d$ variables and an integer $l$, 
             decide whether there is a solution satisfying $l$ of the inequalities.
\end{itemize}

We prove that all three problems are W[1]-hard with respect to $d$. This is done by fpt-reductions from the 
$k$-independent set (or clique) problem in general graphs, which is W[1]-complete~\cite{DF99}.  
As a side-result, we also show that, when restricted to equalities, 
problem (iii) is W[1]-hard with respect to both $l$ and $d$. 
The reductions for problems (i) and (ii) 
are based on a technique pioneered in Cabello et al.~\cite{CGKR08}, see next section. 
With the addition of these two problems this technique shows a generic trait and its potential as  
a useful tool for proving hardness of geometric problems with respect to the dimension.   

In all three reductions the dimension is linear in the size $k$ of the independent set (or clique), hence
an $n^{o(d)}$-time algorithm for any of the problems implies an $n^{o(d)}$-time algorithm for the parameterized $k$-clique  
problem, 
which in turn implies that $n$-variable $3$SAT can be solved in $2^{o(n)}$-time. 
The Exponential Time Hypothesis (ETH)~\cite{DBLP:journals/jcss/ImpagliazzoP01} conjectures that no such algorithm exists.

\paragraph{Related work.}
The dimension of geometric problems is a natural parameter for studying their parameterized complexity.
However, there are only few results of this type: 
Langerman and Morin~\cite{DBLP:journals/dcg/LangermanM05} gave fixed-parameter tractability results for 
the problem of covering points with hyperplanes, while the `dual' parameterization 
of the maximum-size feasible subsystem problem, where parameter $l$ is now the smallest number of inequalities 
one has to remove to make the system feasible 
is fixed-parameter tractable with respect to both $l$ and $d$~\cite{BCILM08}.
As for hardness results, the problems of covering points with balls and computing the volume of 
the union of axis parallel boxes have been shown to be W[1]-hard by  
Cabello et al.~\cite{CGKR08} and 
Chan~\cite{Chan08} respectively. We refer the reader to Giannopoulos et al.~\cite{DBLP:journals/cj/GiannopoulosKW08} for 
a survey on parameterized complexity results for geometric problems.

The problem of stabbing balls in $\Rd$ with one line 
was shown to be NP-hard when $d$ is part of the input by Megiddo~\cite{Megiddo90onthe}.  
This problem is equivalent to the minimum enclosing cylinder problem for points, see Varadarajan et al.~\cite{VVYZ07}. 
Exact and approximation algorithms for the latter problem can be found, for example, 
in B\u{a}doiu et al.~\cite{BHI02}.

Megiddo~\cite{Meg88} showed that the problem of separating two point sets in $\Rd$ by two hyperplanes 
is NP-hard. He also showed that the general problem of separating two point sets by $l$ hyperplanes can 
be solved in polynomial time when both $d$ and $l$ are fixed.

The complexity of the maximum-size feasible subsystem problem was studied in Amaldi and Kann~\cite{AK95}.
Several results on the hardness of approximability can also be found in this paper, as well 
as in Arora et al.~\cite{ABSS97}. For exact and approximation algorithms for this and several related
problems see Aronov and Har-Peled~\cite{AH08}. 

\section{Preliminaries}
\subsection{Methodology}
As mentioned above, 
all three hardness results use a reduction from the $k$-independent set (or clique) problem.
Using the technique in~\cite{CGKR08}, 
we construct of a \emph{scaffolding}
structure that restricts the solutions to $n^k$ combinatorially
different solutions, which can be interpreted as potential
$k$-cliques in a graph with $n$ vertices. Additional \emph{constraint} objects
will then encode the edges of the input graph.

The main ideas are the following. 
We construct geometric instances which lie in Euclidean space whose dimension depends only on $k$. 
Note that the lower the dependence on~$k$, the better the lower bound we get from the hardness result.
In our case the dependence is linear. The scaffolding structure is highly symmetric. 
It is composed of $k$ symmetric subsets of a linear (in $n$) number of objects that lie in orthogonal subspaces. 
Orthogonality together with the specific geometric properties of each problem allows us to 
restrict the solutions to $n^k$ combinatorially different solutions. The way of placing 
the constraint objects is crucial: each object lies in a $4$-dimensional 
subspace and cancels an 
exponential number of solutions.

\paragraph{Model of computation.}
The geometry of the constructions in
Sections~\ref{min_encl_cyl},~\ref{sep_hyper} will be described as if
exact square roots and expressions of the form $\sin\frac{\pi}{n}$
were available.  To make the reduction suitable for the Turing machine
model, the data must be perturbed using fixed-precision roundings.  This
can be done with polynomially many bits in a way similar to the
rounding procedure followed in~\cite{CGKR08,cgkmr-gcfpt-09}.  We omit the details
here.
The construction in Section~\ref{max_feasible} uses small integral
data.
\subsection{Notation}
Let $[n]=\{1,\ldots,n\}$ and $G([n]),E)$ be an undirected graph. 

In sections~\ref{min_encl_cyl},~\ref{sep_hyper}, 
it will be convenient to
view $\Reals^{2k}$ as the product of $k$ orthogonal planes $E_1,\dots, E_k$,
where each $E_i$ has coordinate axes $X_i,Y_i$. The origin is denoted by $o$.
The coordinates of a point $p\in\R^{2k}$ are denoted by $\left(x_1(p), y_1(p),\ldots ,x_k(p), y_k(p)\right)$.
The notions of a point and vector will be used interchangeably. We denote by 
$C_i$ the unit circle on $E_i$ centered at $o$. 

\section{Minimum enclosing cylinder (or stabbing balls with one line)}
\label{min_encl_cyl}
 
Given an undirected graph $G([n], E)$ we construct a set $\mathcal{B}$ of 
balls of equal radius $r$ in $\R^{2k}$ such that 
$\mathcal{B}$ can be stabbed by a line if and only if $G$ has an independent set of size $k$. 

For every ball $B\in\mathcal{B}$ we will also have $-B\in\mathcal{B}$. This allows us to restrict our attention to lines 
through the origin: a line that stabs $\mathcal{B}$ can be translated so that it goes through 
the origin and still stabs $\mathcal{B}$. In this section, by a line we always mean a line through the origin. 
For a line $l$, let $\vec{l}$ be its unit direction vector.

For each plane $E_i$, we define $2n$ $2k$-dimensional balls, 
whose centers $c_{i1},\ldots ,c_{i2n}$ are regularly spaced on the unit circle $C_i$.
Let $c_{iu}\in E_i$ be the center of the ball $B_{iu}$, $u\in[2n]$, with  
$$
x_i(c_{iu})=\cos(u-1)\tfrac{\pi}{n},\; y_i(c_{iu})=\sin(u-1)\tfrac{\pi}{n}.
$$
We define the scaffolding ball set $\mathcal{B}^0=\{B_{iu},\, i=1,\ldots,k \;\mathrm{and}\; u=1,\ldots,2n\}$. 
We have $|\mathcal{B}^0|=2nk$. All balls in $\mathcal{B}^0$ will have the same radius $r<1$, to be defined later.

Two antipodal balls $B$, $-B$ are stabbed by the same set of lines.
A line $l$ stabs a ball $B$ of radius $r$ and center $c$ if and only if $(c\cdot \vec{l})^2\geq \lVert c\rVert^2-r^2$. 
Thus, $l$ stabs $\mathcal{B}^0$ if and only if it satisfies the following system of $nk$ inequalities:
\begin{equation}
\nonumber 
(c_{iu}\cdot \vec{l})^2\geq \lVert c_{iu}\rVert^2-r^2=1-r^2, 
\;\;\text{for}\;\;
i=1,\ldots,k
\;\;\text{and}\;\;
u=1,\ldots,n. 
\end{equation}

Consider the inequality asserting that $l$ stabs $B_{iu}$.
Geometrically, it amounts to saying that
the projection $\vec{l}_i$ of $\vec{l}$ on the plane $E_i$ lies in one of the half-planes
\begin{equation*}
H_{iu}^+=\{p\in E_i | c_{iu}\cdot p \geq \sqrt{\lVert c_{iu}\rVert^2-r^2}\} 
\;\;\mathrm{or}\;\; 
H_{iu}^-=\{p\in E_i | c_{iu}\cdot p \leq -\sqrt{\lVert c_{iu}\rVert^2-r^2}\}.
\end{equation*}
Consider the situation on a plane $E_i$.
Looking at all half-planes $H_{i1}^+, H_{i1}^-,\ldots,H_{in}^+, H_{in}^-$, we see that 
$l$ stabs all balls $B_{iu}$ (centered on $E_i$) if and and only if $\vec{l}_i$ lies in one of the $2n$ wedges
$\pm (H_{i1}^- \cap H_{i2}^+),\ldots, \pm (H_{i(n-1)}^- \cap H_{in}^+), \pm (H_{i1}^- \cap H_{in}^-)$; 
see Fig.~\ref{wedges}. 
\begin{figure}
  \centering
	 \includegraphics[width=7.5cm]{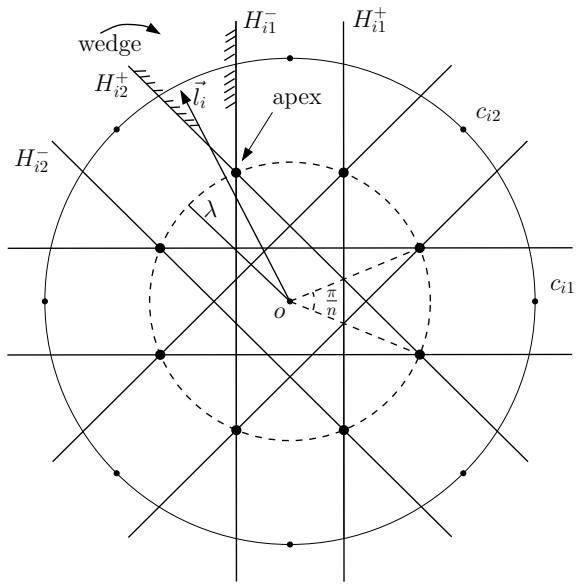}
	 \caption{Centers of the balls and their respective half-planes and wedges on a plane $E_i$, for $n=4$.}
	 \label{wedges}
\end{figure}
The apices of the wedges are regularly spaced on a circle of radius 
$\lambda=\sqrt{2(1-r^2)/(1-\cos\frac{\pi}{n})}$, 
and define the set 
\begin{equation*}
A_i=\{\pm \left(\lambda \cos(2u-1)\tfrac{\pi}{2n}, \lambda \sin(2u-1)\tfrac{\pi}{2n}\right)\in E_i,\, u=1,\ldots,n\}.
\end{equation*}
For $l$ to stab all balls $B_{iu}$, we must have that $\lVert \vec{l}_i\rVert\geq\lambda$.
We choose $r=\sqrt{1-(1-\cos\frac{\pi}{n})/(2k)}$ in order to obtain $\lambda=1/\sqrt{k}$.  

Since the above hold for every plane $E_i$, and since $\vec{l}\in\R^{2k}$ is a unit vector, 
we have
$$
1=\lVert l\rVert^2=\lVert l_1\rVert^2 +\cdots +\lVert l_k\rVert^2 \geq k\lambda^2 = 1.
$$
 Hence, equality holds throughout, which implies that $\lVert\vec{l}_i\rVert=1/\sqrt{k}$, for every $i\in\{1,\ldots,k\}$. Hence, for 
line $l$ to stab all balls in $\mathcal{B}^0$, every projection $\vec{l}_i$ must be one 
of the $2n$ apices in $A_i$.  Each projection $\vec{l}_i$ can be chosen independently. There are $2n$ choices, but since $\vec{l}$ and 
$-\vec{l}$ correspond to the same line, the total number of lines that stab $\mathcal{B}^0$ is $n^k 2^{k-1}$. 

For a tuple $(u_1,\ldots,u_k)\in[2n]^k$, we will denote by $l(u_1,\ldots, u_k)$ the stabbing line with direction vector
\begin{equation*}
\frac{1}{\sqrt{k}}\left(\cos(2u_1-1)\tfrac{\pi}{2n}, \sin(2u_1-1)\tfrac{\pi}{2n}, \ldots, 
\cos(2u_k-1)\tfrac{\pi}{2n}, \sin(2u_k-1)\tfrac{\pi}{2n}\right).
\end{equation*}
Two lines $l(u_1,u_2,...,u_k)$
and $l(v_1,v_2,...,v_k)$ are said to be equivalent
if $u_i \equiv v_i\pmod n$, for all $i$. 
This relation defines $n^k$ equivalence classes $L(u_1,\ldots,u_k)$, with $(u_1,\ldots,u_k)\in[n]^k$, 
where each class consists of $2^{k-1}$ lines. 

From the discussion above, it is clear that there is a bijection between the possible equivalence classes of lines that 
stab $\mathcal{B}^0$ and $[n]^k$. 
\subsection{Constraint balls}

We continue the construction of the ball set $\mathcal{B}$ by showing how
to encode the structure of~$G$.  For each pair of distinct indices $i\neq
j$ ($1\leq i,j \leq k$) and for each pair of (possibly equal) vertices
$u,v\in[n]$, we define a \emph{constraint set} $\mathcal{B}_{ij}^{uv}$ of balls
with the property that (all lines in) all classes $L(u_1,\ldots,u_k)$ stab
$\mathcal{B}_{ij}^{uv}$ except those with $u_i=u$ and $u_j=v$. The centers
of the balls in $\mathcal{B}_{ij}^{uv}$ lie in the $4$-space $E_i\times
E_j$. Observe that all lines in a particular class $L(u_1,\ldots,u_k)$ project onto
only two lines on $E_i\times E_j$.
We use a ball $B_{ij}^{uv}$ (to be defined shortly) of radius~$r$ that is stabbed by \emph{all} lines $l(u_1,\ldots,u_k)$ except those 
with $u_i=u$ and $u_j=v$.
Similarly, we use a ball $B_{ij}^{u\bar v}$ that is stabbed by \emph{all} lines $l(u_1,\ldots,u_k)$ except those 
with $u_i=u$ and $u_j=\bar v$, where $\bar v=v+n$. 
Our constraint set consists then of the four balls
$$\mathcal{B}_{ij}^{uv}=\{\pm B_{ij}^{uv}, \pm B_{ij}^{u\bar v}\}.$$

We describe now the placement of a ball $B_{ij}^{uv}$. Consider a line $l=l(u_1,\ldots,u_k)$ with 
$u_i=u$ and $u_j=v$. The center $c_{ij}^{uv}$ of $B_{ij}^{uv}$ will lie on a line ${z}\in E_i\times E_j$ that is
orthogonal to~$\vec{l}$, but not orthogonal to any line $l(u_1,\ldots,u_k)$ with 
$u_i\neq u$ or $u_j\neq v$. We choose the direction
$\vec{z}$ of $z$ as follows: 
$$x_i(\vec{z})=\mu(\cos\theta_i-3n\sin\theta_i),\, 
y_i(\vec{z})=\mu(\sin\theta_i+3n\cos\theta_i),$$ 
$$x_j(\vec{z})=\mu(-\cos\theta_j-6n^2\sin\theta_j),\, 
y_j(\vec{z})=\mu(-\sin\theta_j+6n^2\cos\theta_j),$$
where $\theta_i=(2u-1)\frac{\pi}{2n}$, $\theta_j=(2u-1)\frac{\pi}{2n}$, and $\mu=1/(9n^2+36n^4+2)$. 
It is straightforward to check that $\vec{l}\cdot\vec{z}=0$.

Let $\omega$ 
be the angle between $\vec{l'}$ and $\vec{z}$. 
We have the following lemma: 
\begin{lemma}
\label{omega_bound}
For any line $l'=l(u_1,\ldots,u_k)$, with $u_i\neq u$ or $u_j\neq v$ the angle $\omega$ 
between $\vec{l'}$ and $\vec{z}$ satisfies $|\cos\omega|>\frac{\mu}{\sqrt{k}}$.
\end{lemma}
\begin{proof}
Without loss of generality we consider a fixed direction $\vec{z}$ where $\theta_i=\theta_j=\frac{\pi}{2n}$ (i.\,e., $u=v=1$).
Consider $\vec{l'}$ with $x_i(\vec{l'})=\cos\theta$, $y_i(\vec{l'})=\sin\theta$, $x_j(\vec{l'})=\cos\phi$, and 
$y_j(\vec{l'})=\sin\phi$, where $\theta=(2u_i-1)\frac{\pi}{2n}$ and $\phi=(2u_j-1)\frac{\pi}{2n}$, with $(u_i,u_j)\neq (1,1)$ and 
$(u_i,u_j)\neq (n+1, n+1)$. After straightforward calculations we have that
$|\cos\omega|=|\vec{l'}\cdot\vec{z}|=\frac{\mu}{\sqrt{k}}|\alpha|$, where 
$$\alpha=\cos(u_i-1)\tfrac{\pi}{n} + 3n\sin(u_i-1)\tfrac{\pi}{n} - \cos(u_j-1)\tfrac{\pi}{n} + 6n^2\sin(u_j-1)\tfrac{\pi}{n}.$$
We will show that $|\alpha|>1$. We will use the inequality:
$$|\sin(u_i-1)\tfrac{\pi}{n}|\geq|\sin\tfrac{\pi}{n}|>\tfrac{1}{n},$$
which holds for all $1\leq u_i\leq 2n$, with $u_i\neq 1$, $u_i\neq n+1$, and $n\geq 4$.
We examine the following cases:

(i) $u_j\neq 1$ and $u_j\neq n+1$. Then $u_i$ can take any value.
We have 
\begin{align*}
|\alpha| &\geq
\left||6n^2\sin(u_j-1)\frac{\pi}{n}|-|\cos(u_j-1)\frac{\pi}{n} - \cos(u_i-1)\frac{\pi}{n} - 3n\sin(u_i-1)\frac{\pi}{n}|\right|
\\&
>|6n^2\cdot \tfrac{1}{n} - |2+3n||
\\&
=3n-2 > 1.
\end{align*}

(ii) $u_j=1$. Then $u_i\neq 1$. If also $u_i\neq n+1$, we have
\begin{align*}
|\alpha| &\geq
|0-1+3n\sin(u_i-1)\frac{\pi}{n} + \cos(u_i-1)\frac{\pi}{n}|
\\&
>|-1+3n\cdot \frac{1}{n}-1| = 1.
\end{align*} 
If $u_i=n+1$, then $|\alpha|=2$.

(iii) $u_j=n+1$. Then $u_i\neq n+1$. The two cases where $u_i\neq 1$ or $u_i=1$ are dealt with similarly to the previous case.
\end{proof}
This lower bound on $|\cos\omega|$ helps us 
place $B_{ij}^{uv}$ sufficiently close to the origin so that it is still intersected by $l'$, i.\,e., 
$\vec{l'}$ lies in one of the half-spaces 
$c_{ij}^{uv}\cdot p\geq \sqrt{\lVert c_{ij}^{uv}\rVert^2-r^2}$ or $c_{ij}^{uv}\cdot p\leq -\sqrt{\lVert c_{ij}^{uv}\rVert^2-r^2}$, 
$p\in\R^{2k}$.

We claim that any point $c_{ij}^{uv}$ on $z$ with $r<\lVert c_{ij}^{uv}\rVert<\sqrt{\frac{k}{k-\mu^2}}r$ will do.
For any position of $c_{ij}^{uv}$ on $z$ with $\lVert c_{ij}^{uv}\rVert>r$, we have 
$(c_{ij}^{uv}\cdot \vec{l})^2=0<\lVert c_{ij}^{uv}\rVert^2-r^2$, 
i.\,e., $l$ does not stab $B_{ij}^{uv}$. 
On the other hand, as argued above we need that $|c_{ij}^{uv}\cdot \vec{l'}|\geq \sqrt{\lVert c_{ij}^{uv}\rVert^2-r^2}$.
Since $c_{ij}^{uv}\cdot \vec{l'}=\cos\omega\cdot \lVert c_{ij}^{uv}\rVert$, 
we have the condition $|\cos\omega|\geq\sqrt{1-\frac{r^2}{\lVert c_{ij}^{uv}\rVert^2}}$. By Lemma~\ref{omega_bound} we know that 
$|\cos\omega|>\frac{\mu}{\sqrt{k}}$, hence by choosing $\lVert c_{ij}^{uv}\rVert$ 
so that $\frac{\mu}{\sqrt{k}}>\sqrt{1-\frac{r^2}{\lVert c_{ij}^{uv}\rVert^2}}$ 
we are done.
\paragraph{Reduction.}
Similarly to~\cite{CGKR08}, the structure of the input graph $G([n], E)$ can now be represented as follows.
We add to $\mathcal{B}^0$ the $4n\binom k 2$ balls in 
$\mathcal{B}_V=\bigcup \mathcal{B}_{ij}^{uu},\, 1\leq u\leq n, \, 1\leq i<j\leq k$, to ensure that 
all components $u_i$ in a solution (class of lines $L(u_1,\ldots,u_k)$) are distinct. 
For each edge $uv\in E$ we also add the balls in $k(k-1)$ sets $\mathcal{B}_{ij}^{uv}$, with $i\neq j$.
This ensures that the remaining classes of lines $L(u_1,\ldots,u_k)$ represent independent sets of size $k$.
In total, the edges are represented by the $4k(k-1)|E|$ balls in 
$\mathcal{B}_E=\bigcup \mathcal{B}_{ij}^{uv},\, uv\in E,\, 1\leq i,j\leq k,\, i\neq j$.
The final set $\mathcal{B}=\mathcal{B}^0\cup\mathcal{B}_V\cup\mathcal{B}_E$ has $2nk+4\binom k 2 (n+2|E|)$ balls.

As noted in above, there is a bijection between the possible equivalence classes of lines $L(u_1,\ldots,u_k)$ that 
stab $\mathcal{B}$ and the tuples $(u_1,\ldots,u_k)\in [n]^k$. The constraint sets of balls exclude tuples with two equal indices 
$u_i=u_j$ or with indices $u_i$, $u_j$ when $u_iu_j\in E$, thus, the classes of lines that stab $B$ represent exactly the independent 
sets of $G$. Thus, we have the following:
\begin{lemma}
Set $\mathcal{B}$ can be stabbed by a line if an only if $G$ has an independent set of size $k$.
\end{lemma}
From this lemma and since this is an fpt-reduction, we conclude:
\begin{theorem}
Deciding whether $n$ unit balls in $\R^{d}$ can be stabbed with one line is 
\textup{W[1]}-hard with respect to $d$.
\end{theorem}

\section{Separating two point sets by two hyperplanes}
\label{sep_hyper}

Let $P$ and $Q$ be two point sets in $\Rd$.
Two hyperplanes split space generically into four ``quarters''. 
There are three different versions of what it means to separate $P$ and $Q$ by two hyperplanes: 

\begin{itemize}\addtolength{\itemsep}{-0.5\baselineskip}
\item[(a)] Each quarter contains only points of one set.
\item[(b)] The set $Q$ is contained in one quarter only, 
           and set $P$ can populate the remaining three quarters.
\item[(c)] Same as (b), but the roles of $P$ and $Q$ are not fixed in advance.
\end{itemize}

In the following we work only with version (a), which is the most general. 
For the point sets that we construct,
it will turn out that if a separation according to (a)
exists, it will also be valid by (b) and (c).
Thus, our reduction works for all three versions of the problem.

Separation according to (a) 
is equivalent to requiring that
every segment $pq$ between a point $p\in P$ and a point $q\in Q$ is intersected by one of the two hyperplanes.
Note that we restrict our attention to \emph{strict} separation, i.\,e., no 
hyperplane can go through a point of $P$ or $Q$. (The result extends to weak separation; see 
the end of this section.)

Given an undirected graph $G_0([n_0], E_0)$ with $n_0\geq2$ and an
integer $k$, we construct two point sets $P$ and $Q$ in $\R^{2k}$ with
the property that they can be separated by two hyperplanes if and only
if $G_0$ has an independent set of size $k$.  For technical reasons,
we duplicate the vertices of the graph: we build a new graph with $n=
2n_0$ vertices.  Every vertex $u\in[n_0]$ of the original graph gets a
second copy $u' := u+n_0$, and for every original edge $uv$, there are
now four edges $uv$, $uv'$, $u'v$, $u'v'$. The new graph $G([n],E)$
has an independent set of size $k$ if and only if the original graph
has such a set.

On each plane $E_i$, $i=1,\ldots ,k$, we define a set $P_i$ of $n$ points regularly spaced on the circle $C_i$:
$$
	P_i=\{\,p_{iu}\in E_i \mid x_i(p_{iu})=\cos (u-1)\tfrac{2\pi}{n},\,
		y_i(p_{iu})=\sin (u-1)\tfrac{2\pi}{n},\ u=1,\ldots ,n\,\}.
$$
For an index $u\in [n]$, it will be convenient to define its \emph{antipodal} and \emph{almost antipodal} partner 
$u'=u+\frac{n}{2}$ and $\bar u=u+\frac{n}{2} + 1$ respectively. (All
indices are modulo $n$). Thus we are extending the notation $u'$ to
all (original and new) vertices $u$, with $(u')'=u$.

The scaffolding is defined by two sets $P=\bigcup P_i$ and $Q^0=\{o\}$. We have $|P|=nk$. 

Since the points in each $P_i$ are regularly spaced on $C_i$, 
a hyperplane that does not contain the origin can intersect at most $n/2$ segments $op_{iu}$ on each plane $E_i$. Hence, 
at least two hyperplanes are needed to separate $P$ and $Q^0$. Actually, two suffice.
One hyperplane can intersect the $n/2$ consecutive (in a counter-clockwise order) segments 
$op_{i\bar u_i},\ldots ,op_{i u_i}$ on each $E_i$, for a choice of $u_i\in [n]$ (see Fig.~\ref{separation}). 
\begin{figure}
  \centering
	 \includegraphics[width=9cm]{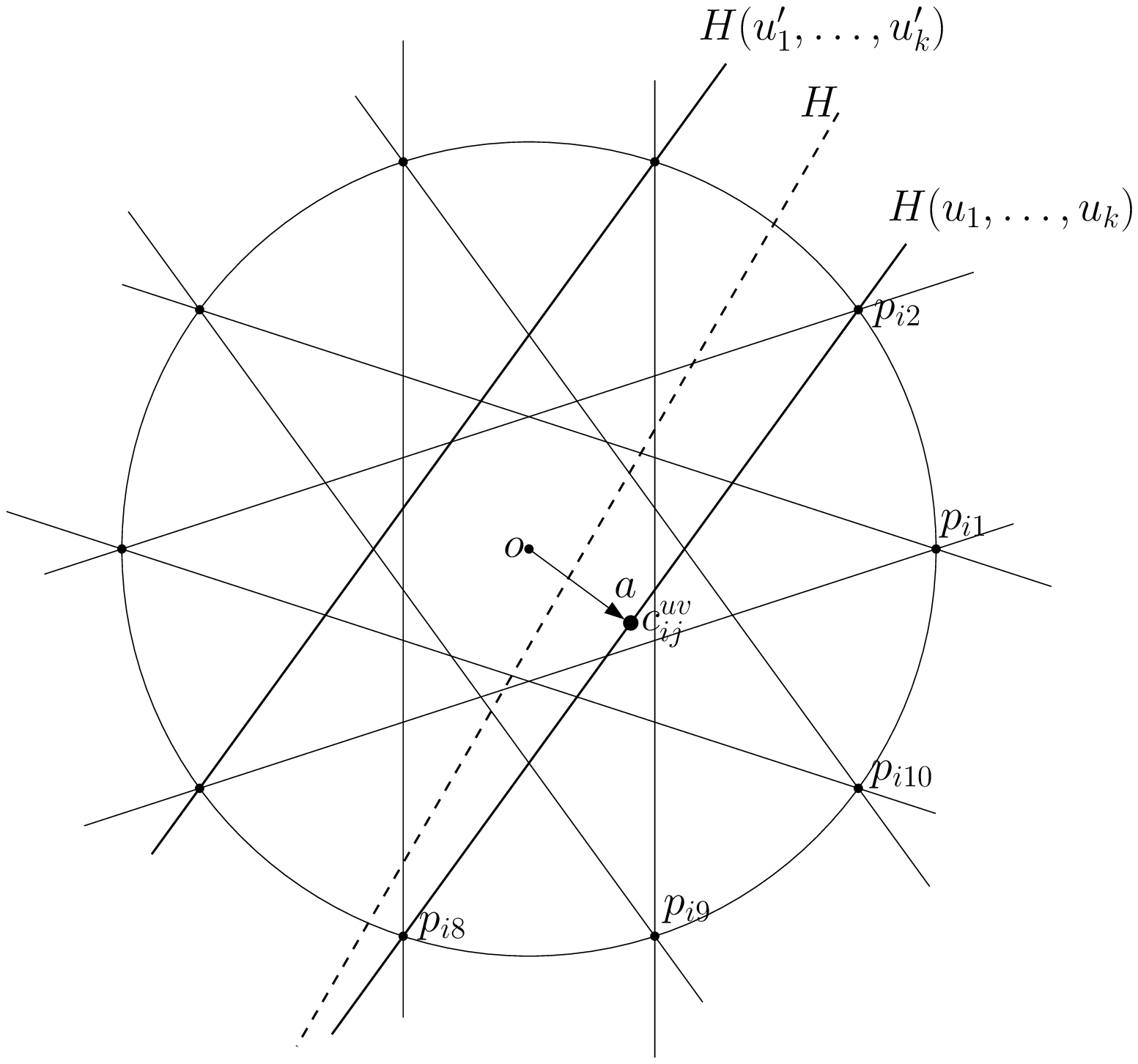}
	 \caption{Point set $P_i$, for $n=10$, 
                  a hyperplane $H$ in the class $\mathcal{H}(u_1,\ldots ,u_k)$} and the corresponding 
                  boundary hyperplane $H(u_1,\ldots ,u_k)$ for $u_i=2$. The placement of $q_{ij}^{uv}$ 
                  is shown in a two-dimensional analog.
	 \label{separation}
\end{figure}
There is an infinite number of such hyperplanes, 
forming an equivalence class $\mathcal{H}(u_1,\ldots ,u_k)$. 
Since the planes $E_1,\ldots ,E_k$ are orthogonal, 
each $u_i$ independently defines which of the $n/2$ consecutive segments on $E_i$ are intersected by a hyperplane 
in $\mathcal{H}(u_1,\ldots ,u_k)$.
The remaining $n/2$ segments 
$-op_{i\bar u_i},\ldots ,-op_{i u_i}$ on each $E_i$ 
can then be intersected by any hyperplane in the `complementary' class 
$\mathcal{H}(u'_1,\ldots ,u'_k)=\{-H \mid H\in \mathcal{H}(u_1,\ldots ,u_k)\}$.
Effectively, every hyperplane in $\mathcal{H}(u_1,\ldots ,u_k)$ separates $Q^0$ from the $\frac{kn}{2}$-point set
$P(u_1,\ldots ,u_k)=\{p_{1\bar u_1},\ldots ,p_{1 u_1}\}\cup\cdots \cup \{p_{k\bar u_k},\ldots ,p_{k u_k}\}$.
Concluding, there are $n^k$ possible partitions of 
$P$
into two groups, each separated from $Q^0$ by one hyperplane,
in correspondence to the $n^k$ possible tuples $(u_1,\ldots ,u_k)\in [n]^k$:
\begin{lemma}\label{lemma-candidate-hyperplanes}
  The possible pairs of hyperplanes
that separate
$P$ from $Q^0$ are of the form  $h,h'$ with
 $h\in\mathcal{H}(u_1,\ldots ,u_k)$ and
$h'\in\mathcal{H}(u'_1,\ldots ,u'_k)$, for some
 $(u_1,\ldots ,u_k)\in [n]^k$.  
\end{lemma}
Since by construction, the graph $G$ has the property that $uv\in E$
iff $u'v'\in E$, the separating pairs of hyperplanes $h,h'$ can be
used to encode the potential independent sets $\{u_1,\ldots ,u_k\}$:
it does not matter which of $h$ and $h'$ we choose, the corresponding
vertex set will be an independent set in both cases, or a dependent
set in both cases.

\subsection{Constraint points}\label{sep_constraint_points}

For each pair of indices $i\neq
j$ ($1\leq i,j \leq k$) and for each pair of (possibly equal) vertices
$u,v\in[n]$, we will define a constraint point $q_{ij}^{uv}\in E_i\times E_j$
with the following property: in every class $\mathcal{H}(u_1,\ldots ,u_k)$, there is a hyperplane that 
separates $\{q_{ij}^{uv}\}$ from $P(u_1,\ldots ,u_k)$ 
except those classes with $u_i=u$ and $u_j=v$ (in which case no such hyperplane exists). 
In this way, no partition of $P$ into sets $P(u_1,\ldots ,u_k)$ and $P(u'_1,\ldots ,u'_k)$ 
with $u_i=u$ and $u_j=v$ will be possible such that each set is separated from $Q^0\cup\{q_{ij}^{uv}\}$ by 
a hyperplane. 

Let $H(u_1,\ldots ,u_k)$ be the unique hyperplane 
through the $2k$ affinely independent points 
$p_{1 u_1}, p_{1\bar u_1},\ldots ,p_{k u_k}, p_{k\bar u_k}$. Note that 
$H(u_1,\ldots ,u_k)$ is \emph{not} in the class $\mathcal{H}(u_1,\ldots ,u_k)$, since we want strict separation; 
informally,
$H(u_1,\ldots ,u_k)$ lies at the boundary of $\mathcal{H}(u_1,\ldots ,u_k)$, with an appropriate parameterization of hyperplanes:
moving
$H(u_1,\ldots ,u_k)$ towards the origin by a sufficiently small amount leads to a hyperplane in $\mathcal{H}(u_1,\ldots ,u_k)$.

We define the constraint point $q_{ij}^{uv}$ as
the centroid of the four points
 $p_{i u}, p_{i\bar u} ,p_{j v}, p_{j\bar v}$.
Its nonzero coordinates are
$$
 x_i=
\frac{\cos\theta_i + \cos\bar\theta_i}4
,\
 y_i=
\frac{\sin\theta_i + \sin\bar\theta_i}4
,\
 x_j=
\frac{\cos\theta_j + \cos\bar\theta_j}4
,\
 y_j=
\frac{\sin\theta_j + \sin\bar\theta_j}4
,
 $$
for $\theta_i=(u-1)\frac{2\pi}{n}$, $\bar\theta_i=(\bar u-1)\frac{2\pi}{n}$,
 $\theta_j=(v-1)\frac{2\pi}{n}$, and $\bar\theta_j=(\bar v-1)\frac{2\pi}{n}$.

\begin{lemma}\label{lemma1_sep_hyper}
  If $u_i=u$ and $u_j=v$, no hyperplane in 
$\mathcal{H}(u_1,\ldots
  ,u_k)$ separates $q_{ij}^{uv}$ from
$P(u_1,\ldots ,u_k)$.
\end{lemma}
\begin{proof}
  Such a hyperplane would in particular have to separate
$q_{ij}^{uv}$ from $p_{i u}, p_{i\bar u} ,p_{j v}, p_{j\bar v}$,
which is impossible.
\end{proof}

To see that 
$q_{ij}^{uv}$ does not ``destroy'' the classes
 $\HH(u_1,\ldots ,u_k)$ with $u_i\neq u$ or $u_j\neq v$,
let us consider a fixed pair of indices $i\neq
j$.
 All points
$q_{ij}^{uv}$, ($u,v\in[n]$) lie on a sphere $S_{ij}$
 around the origin in $E_i\times E_j$
(of radius $\sqrt {1/2}\cdot\sin \frac\pi n$).
 The intersection 
$H(u_1,\ldots ,u_k)\cap (E_i\times E_j)$ is a 
3-dimensional hyperplane
 $F_{ij}^{u_i u_j}$ uniquely defined by $u_i$ and $u_j$:
 $F_{ij}^{u_i u_j}$ goes through
the four points $p_{i u_i}, p_{i\bar u_i} ,p_{j u_j}, p_{j\bar u_j}$.
Moreover, $q_{ij}^{u_i u_j}$ is the point where 
 $F_{ij}^{u_i u_j}$ touches the sphere $S_{ij}$.
(This follows from symmetry considerations, and it can
also be checked by a straightforward calculation that the vector
 $q_{ij}^{u_i u_j}$ is perpendicular to the hyperplane
 $F_{ij}^{u_i u_j}$.)
This allows us to conclude:
\begin{lemma}\label{lemma2_sep_hyper}
If $u_i\neq u$ or $u_j\neq v$, 
then
$q_{ij}^{uv}$ lies on the same side of
 the hyperplane $H(u_1,\ldots ,u_k)$ as the origin~$o$.
\end{lemma}
\begin{proof}
  The point $q_{ij}^{uv}$ lies on the sphere $S_{ij}\in E_i\times E_j$
  centered at the origin.  This sphere lies on the same side of
  $H(u_1,\ldots ,u_k)$ as the origin, except for the point where it
  touches $H(u_1,\ldots ,u_k)$. But this touching point $q_{ij}^{u_i
    u_j}$ is different from $q_{ij}^{uv}$.
\end{proof}

 This means that $q_{ij}^{uv}$ and the points in $P(u_1,\ldots ,u_k)$ are 
on different sides of the hyperplane $H(u_1,\ldots ,u_k)$ (except for the points 
$p_{1 u_1}, p_{1\bar u_1},\ldots ,p_{k u_k}, p_{k\bar u_k}$, which lie on it). Since $q_{ij}^{uv} \notin H(u_1,\ldots ,u_k)$, 
every sufficiently close translate of $H(u_1,\ldots ,u_k)$ in $\mathcal{H}(u_1,\ldots ,u_k)$ with $u_i\neq u$ or $u_j\neq v$
separates $P(u_1,\ldots ,u_k)$ and $\{q_{ij}^{uv}\}$.



\paragraph{Reduction.}
Similarly to the reduction in Section~\ref{min_encl_cyl}, 
we encode the structure of $G$ by adding to $Q^0$ the $n\binom{k}{2}$ constraint points $q_{ij}^{uu}$
($1\leq u\leq n, \, 1\leq i<j\leq k$)
and $2|E|\binom{k}{2}$ constraint points $q_{ij}^{uv}$ ($uv\in E$ and $i\neq j$).
Let $Q$ be the resulting point set. 
Then the possible partitions of $P$ into two sets, each separated from $Q$ by one hyperplane, represent 
the independent sets of $G$.
\begin{lemma}
Sets $P$ and $Q$ can be separated by two hyperplanes if and only if $G$ has an independent set of size $k$.
\end{lemma}
From this lemma, and since this is an fpt-reduction, we conclude with the following:
\begin{theorem}
Deciding whether two point sets $P, Q$ in $\Rd$ can be separated by two hyperplanes is \textup{W[1]}-hard with respect to $d$.
\end{theorem}

\paragraph{Remark.}
The construction above depends on requiring
strict separation, i.\,e., the separating hyperplanes are not allowed to go through the given points.
For the fixed-precision approximation that is necessary 
to make the reduction suitable for a Turing machine,
we have to move the constraint points
$q_{ij}^{uv}$ a little bit further away from the center
before rounding them to rational coordinates.
The statement of Lemma~\ref{lemma1_sep_hyper} is refined
and excludes the possibility of separating
$P(u_1,\ldots ,u_k)$ from 
the set $\{o,q_{ij}^{uv}\}$ rather than from the point
$q_{ij}^{uv}$ alone.

These modifications are also suitable for the version of the problem
where \emph{weak separation} is allowed, i.\,e., points on the separation
boundary can be from $P$ or $Q$ arbitrarily. In this case
$\binom{2k}2$ additional points on the coordinate planes close to the
origin must be added to $Q^0$, in order to eliminate the coordinate
hyperplanes as potential separating hyperplanes.

\section{Maximum-size feasible subsystem}\label{max_feasible}

We first consider the special problem: Given a system of linear equations 
find a solution that satisfies as many equations as possible. 
(Note that this problem is dual to the problem of covering as many points as possible by a 
hyperplane through the origin.)
The decision version of this problem is as follows: Given a set of $n$ hyperplanes in $\Rd$ and an integer~$l$,
decide whether there exists a point in $\Rd$ that is covered by at least $l$ of the hyperplanes. 

In the following, $\vx = (x_1,\dots,x_k)\in\R^k$ denotes a $k$-dimensional
vector (a notation that is slightly different from the one used in the previous sections).
We identify the grid $[n]^k$ with the set of vectors in $\R^k$ with integer coordinates in $[n]$.

For a set $\HH$ of hyperplanes in $\R^k$ and a point $\vx \in\R^k$ we define
$$\depth(\vx,\HH) = |\{h\in \HH \mid \vx\in h\}|.$$

Given an undirected graph $G([n], E)$ and $k \in\N$, we will now construct 
a set $\HH_{G,k}$ of $nk+2|E|\binom{k}{2}$ 
hyperplanes in $\R^k$ such that $G$ has a clique
of size $k$ if and only if
   there is a point $\vx\in\R^k$ with
 $\depth(\vx, \HH_{G,k})=k+\binom{k}{2}$.

For $1 \le i \le k$ and $1 \le v \le n$ we define the hyperplane
$h^v_{i} = \{\vx \mid x_i=v\}$. The scaffolding hyperplane set $ \HH^0 = \{\,h^v_{i} \mid
1 \le i \le k,\ 1 \le v \le n\,\}$ consists of $nk$ hyperplanes.
 Any point $\vx$ is contained in at most $k$ hyperplanes in $\HH^0$;
equality is realized for the points in $[n]^k$:

\begin{lemma}\label{lem:depth_var}
$\depth(\vx,\HH^0)\leq k$ for any $\vx\in\R^k$, 
and $\depth(\vx,\HH^0)=k$ if and only if $\vx\in [n]^k$.
\end{lemma}

For $1 \le i <j \le k$ and $1 \le u, v \le n$ we define the hyperplane
$h_{ij}^{uv} = \{\,\vec x \mid (x_i-u)+n(x_j-v)=0 \,\}$. 
This hyperplane contains only those points $\vx$ of the grid for which $x_i=u$ and $x_j=v$:

\begin{lemma}\label{lem:1}
  $\vx \in h^{uv}_{ij} \cap [n]^k$ if and only if $x_i=u$ and $x_j=v$.
\end{lemma}

\begin{proof}
Assume  $\vx \in h^{uv}_{ij} \cap [n]^k$, i.\,e. 
 $ (x_i-u)+n(x_j-v)=0 $ and $x_i,x_j\in[n]$.
If $x_i\ne u$, the left-hand side of the equation
is not divisible by $n$ and thus cannot 
be~0. Therefore, $x_i=u$ and thus, $x_j=v$. The other direction is obvious.
\end{proof}



For $1 \le i < j \le k$ we define the set
$ \HH^{E}_{ij} = \{\,h^{uv}_{ij} \mid uv\in E\; \mathrm{or}\; vu\in E\,\}$ 
of $2|E|$ hyperplanes. All these hyperplanes are parallel; 
thus a point is contained in at most one hyperplane of $\HH^{E}_{ij}$.
By Lemma~\ref{lem:1}, a point $\vx\in [n]^k$ is contained in a hyperplane of $\HH^{E}_{ij}$ 
if and only if $x_ix_j$ is an edge of $E$.

We define the set $ \HH^{E} = \bigcup_{ 1 \le i < j \le k} \HH^{E}_{ij}$
consisting of $2|E|\binom{k}{2}$ hyperplanes. From the above, we have the following facts:


\begin{lemma}\label{lem:depth_HHE}
\begin{itemize}\addtolength{\itemsep}{-0.5\baselineskip}
\item[\rm(a)] $\depth(\vx,\HH^E)\leq \binom{k}{2}$ for any $\vx\in\R^k$.
\item[\rm(b)] Let $\vx \in [n]^k$. Then $ \depth(\vx,\HH^E) =  |\{\, (i,j) \mid 1 \le i < j \le k,\ x_ix_j \in E \,\}|$ 
\item[\rm(c)] Let $\vx \in [n]^k$. Then $ \depth(\vx,\HH^E) = \binom{k}{2}$ iff $\{x_1,\dots,x_k\}$ is a $k$-clique in $G$.
\end{itemize}
\end{lemma}

For the set $\HH_{G,k} = \HH^0 \cup \HH^{E}$,
 Lemmas~\ref{lem:depth_var} and~\ref{lem:depth_HHE}
 immediately imply:

\begin{lemma}\label{lem:main}
  $ \depth(\vx,\HH_{G,k}) = k+\binom{k}{2}$ if and only if $\vx \in [n]^k$ and $\{x_1,\dots,x_k\}$ is a
  $k$-clique in $G$.
\end{lemma}

Note that the above
construction of the set $\HH_{G,k}$
 is an fpt-reduction with respect to both the depth of the set of hyperplanes, i.\,e., 
the maximum number of hyperplanes covering any point, and the dimension. Hence, we have the following:
\begin{theorem}
Given a set of $n$ of linear equations on $d$ variables and an integer $l$, deciding whether 
there exists a solution that satisfies $l$ of the equations is \textup{W[1]}-hard with respect to both $l$ and $d$.
\end{theorem}

Replacing each equation by $2$ inequalities, 
an instance of the above problem
is transformed into an instance with linear inequalities such that there exists a solution 
satisfying $l$ out of the $n$ equations of the original instance if and only if there exists a solution satisfying $n+l$ out of 
the $2n$ inequalities of the final instance; the number of variables stays the same. Hence, we have the following:

\begin{theorem}
Given a set of $n$ linear inequalities on $d$ variables and an integer $l$, deciding whether 
there exists a solution that satisfies $l$ of the inequalities is \textup{W[1]}-hard with respect to~$d$.
\end{theorem}


\end{document}